\documentclass[preprint,12pt,authoryear]{elsarticle}

\usepackage{amssymb}

\usepackage{mathtools}
\usepackage{hyperref}

\usepackage{amsthm}
\usepackage{tikz-cd}
\newtheorem{theorem}{Theorem}[section]

\newtheorem{definition}[theorem]{Definition}

\journal{Geometry and Physics}

\begin{document}

\begin{frontmatter}

\title{Relating a Geroch-like boundary and the $a$ boundary constructions for spacetimes}

\author[label1]{Abbas M. Sherif}
\ead{abbasmsherif25@gmail.com}
 \address[label1]{Cosmology and Gravity Group, Department of Mathematics and Applied Mathematics, University of Cape Town, Rondebosch 7701, South Africa}
\author[label2]{Gareth Amery}
\ead{garethamery@gmail.com}
 \address[label2]{Astrophysics and Cosmology Research Unit (ACRU), School of Mathematics, Statistics and Computer Science, University of KwaZulu-Natal, Private Bag X54001, Durban 4000, South Africa}

\begin{abstract}
We construct a Geroch-like boundary (when restricted to geodesic curves, this boundary contains as a subset the Geroch's $g$ boundary), which we denote by $\tilde{g}$, and establish an explicit embedding of the $\tilde{g}$ boundary into the $a$ boundary of Scott and Szekeres. This construction, and subsequently the explicit embedding, is done in a 'natural' way (the emphasis on the word natural here will be clarified in the text), thereby answering in the affirmative the outstanding question as to whether there exists a natural way to relate the $g$ and the $a$ boundary constructions.
\end{abstract}


\end{frontmatter}

\section{Introduction}
\label{}

Not long after the field equations were introduced by Einstein, solutions were found with singular points where known physical laws break down. Some of these were only apparently singular points and were actually due to the choice of coordinates. It is possible to choose a set of coordinates that \textit{maximally extend} the spacetime, where the metric becomes analytic through these points (eg. the Kruskal embedding, \cite{kru1}). Truly singular points are those which maximally analytic extensions can not remove. As open sets of the topology on such spacetime manifolds do  not contain these singular points, interpreting such solutions, as well as the effects of the singular points on the solutions, has been a daunting task for mathematical physicists. Over the past sixty years, schemes have been developed, \cite{ger1,ss1,bb1,cc1}, which ``capture'' these points in constructed boundaries. These boundaries are then attached to the spacetime manifold, a process known as manifold completion. The resulting manifold is called the manifold with boundary.

An early investigation of geodesic completeness was carried out by Szekeres, \cite{szk1}, using a power series expansion of the Schwarzschild solution around the coordinate singularity \(r=2m\) to obtain a transformation under which \(r=2m\) is a regular manifold point. \cite{ger1}, provided the first boundary construction for singular spacetimes; the \(g\) boundary. This involves quotienting the set of incomplete geodesics in the spacetime under the equivalence that the geodesics limit to the same point. The \(g\) boundary is restricted in terms of the curves considered because it assumes a definition of a singularity in terms of geodesic incompleteness. However, it is possible to have a non-singular geodesically incomplete spacetime, \cite{ss1,haw1,mis1}. This prompted the construction of alternate boundaries.

\cite{bb1}, constructed the \(b\) boundary. This construction maps endpoints of incomplete curves in the bundle of frames of \(M\), \(\mathcal{L}\left(M\right)\) - generated via Cauchy completion - as additional points of \(M\). It is best known for its application by Hawking and others to singularity theorems, \cite{haw1}, and considers a broader class of curves than the \(g\) boundary. However, this boundary construction is fraught with undesirable properties. For one, a \(4\)-dimensional manifold has a \(20\)-dimensional frame bundle. Though it was shown, \cite{bb1}, that it is sufficient to consider only the bundle of orthonormal frames - which is \(10\)-dimensional - intuition and computations are problematic. As a result, only situations with sufficient symmetries to reduce the dimensions can reasonably be approached by this scheme. Furthermore, there is the problem of the \(b\) boundary being non-Hausdorff which means that a \(b\) boundary point may be arbitrarily close to manifold points.

In 1972 Geroch \textit{et al.}, \cite{cc1}, constructed the \(c\) boundary. This construction is made solely using the causal structure of the spacetime manifold. It attaches future (respectively, past) endpoints to every inextensible curve in the spacetime manifold. These endpoints are called ideal points, and their collection can be interpreted as the boundary at conformal infinity. The attachment is done in such a way that the ideal points only depend on the past (resp., future) of the curves. These ideal points are represented by \textit{terminal indecomposable future} (resp., \textit{past}) sets, \cite{lo4}, which are the maximal future (resp., past) sets that are not the chronological future (resp., past) of any manifold point and cannot be represented as the union of two proper subsets, both of which are open future (resp., past) sets.

As with the other constructions, the \(c\) boundary construction is problematic and carries a lot of compexities. The original topology constructed by Geroch \textit{et al.} has topological separation problems as well as non-intuitive results for certain solutions. As such there have been several modifications of the \(c\) boundary via the construction of different topologies (see references \cite{lo4,as1,bs1,sz1,sh1,sh2}). 

Geroch, Canbin and Wald constructed an example in 1982, \cite{glw1}, for which any singular spacetime boundary construction which falls in certain class (including the \(g\) and \(b\) boundaries) will exhibit pathological topological properties. This led to their assertion that this may in general be true; one might always be able to construct an example for which a particular singular boundary construction fails.

In 1994, Scott and Szekeres constructed the abstract boundary or simply the \(a\) boundary, \cite{ss1}. Given  the existence of open embeddings \(\phi_i:M\longrightarrow \hat{M}_i\), the \(a\) boundary is constructed by defining an equivalence relation on the topological boundaries \(\partial_{\phi_i}M\) with respect to the \(\hat{M}_i\). This construction stands out as it considers a much broader class of curves relative to the other constructions, as well as by virtue of being \(T_2\) separable, which is a crucial requirement for a physically reasonable spacetime, \cite{ss2}. Its generality means that it is applicable in contexts in which one has only an affine connection; for example, Yang-Mills and Einstein-Cartan theories, \cite{ss1}. This construction also turns out to be easier to implement in general compared to the other construction schemes. However, the \(a\) boundary is also plagued whith its own shortcomings (see the discussions in references \cite{ek1,ek2}). It would be nice to have a causal characterization of the \(a\) boundary, as well as being able to define differentiable and metric structures on the \(a\) boundary, and to establish what (if any) new information can be learned from these properties on the \(a\) boundary.

Curiel has argued that the use of curve incompleteness is an adequate definition in context of singular spacetimes and singular structures as far as existing spacetimes of physical relevance are concerned, \cite{ek1,ek2}. This would seem to suggest that the \(g\), \(b\) and \(a\) boundaries are all acceptable boundary constructions in a general relativistic context. This is not necessarily true outside of this immediate context, and this is where the \(a\) boundary - which is more broadly applicable - becomes the construction of choice. Nevertheless, the ubiquity of the application of the curve incompleteness approach motivates a better understanding of the relationship between such constructions and the \(a\) boundary.

To date, there have been no successful results relating the various boundary constructions, although there has been some work on relating the notions of  \(b\)- completeness and \(g\)- completeness, \cite{haw1}. Moreover, within different boundary constructions, there have been several modifications to deal with some of the problems faced by these constructions, \cite{cc1,lo4,bs1,sz1,sh1,sh2}. However, there does not seem to be a ``natural'' and general way to map between different boundary constructions, \cite{as1}.

\subsection{Objective}
The aim of this paper is to construct a Geroch-like boundary - which when restricted to geodesics contains the Geroch \(g\) boundary construction - and explore the relationship between this construction and the \(a\) boundary, which was identified as an important open problem in the original paper constructing the \(a\) boundary \cite{ss1}. If structures on the \(g\) boundary are carried over to this Geroch-like boundary, then such relationship could in principle allow one to construct on the \(a\) boundary, the various structures of interest on the \(g\) boundary. 

We emphasize here that this work does not attempt to address the major limitation of boundary constructions - like the \(b\) boundary, the \(g\) boundary and the \(a\) boundary - where the associated topologies (for the \(a\) boundary we specify here the attached point topology) fail to relate boundary points to points in the interior of the spacetime manifold. This was the motivation for the construction of the strongly attached point topology by \cite{ss10}. In principle our approach could be adapted to the Barry and Scott's strongly attached point topology, but we expect to address this separately.

\subsection{Structure}
In section \ref{02}, we give an overview of the \(g\) boundary and the \(a\) boundary constructions. Section \ref{03} presents our new construction and its relationship to the \(a\) boundary. In section \ref{gtr}, we discuss the limitations of the current work, and what possible route might be taken in future works to address these limitations. In section \ref{05} we summarize our results and anticipate additional future avenues of research.

\section{Some notes on the \(g\) and \(a\) boundaries}\label{02}

 In this section we provide some background on the \(g\) and \(a\) boundary constructions, as well as the topologies that can be placed on these constructions. We follow the standard references \cite{ger1, ss1, haw1, as1,ss2}. 

\subsection{The \(g\) boundary}\label{nnjjt}
Let \(M\) be a geodesically complete spacetime manifold, and let \(\Sigma\) be a co-dimension \(1\) submanifold of \(M\) such that \(M=M_1\sqcup\Sigma\sqcup M_2\). In other words, \(\Sigma\) divides \(M\) into disjoint subsets \(M_1\) and \(M_2\). Suppose one is given one of the subsets, say \(M_1\). A natural question arises as to how much information about \(\Sigma\) can be obtained from what we know about \(M_1\). The idea is to use the information about the incomplete geodesics in \(M_1\) - those in \(M_1\) which, when extended in \(M\), pass through \(\Sigma\) - to recover ``parts'' of \(\Sigma\). One then groups these geodesics as follows: Let \(\gamma\) be an incomplete geodesic in \(M_1\) and generate a family of geodesics by allowing small variations in the initial conditions - a point \(p\in \gamma\) and a tangent vector at \(p\) - of \(\gamma\). This family of geodesics traces out a \(4\)-dimensional tube called a ``thickening'' of \(\gamma\). Another incomplete geodesic \(\gamma'\) is said to be related to \(\gamma\) if \(\gamma'\) enters and remains in every thickening of \(\gamma\). However, since finding \(\Sigma\) may not always suffice, the above construction must be generalized (these notes follow those of reference \cite{ger1}).

Let \(G=TM \setminus\lbrace{0\rbrace}\) be the reduced tangent bundle on a spacetime manifold \(M\), made up of the non-zero vectors in \(TM\). The set \(G\subset TM\) is an \(8\)-dimensional manifold whose elements are the pairs \(\left(p,\xi^{\alpha}\right)\). Each \(\left(p,\xi^{\alpha}\right)\in G\) uniquely determines the geodesic \(\gamma\), satisfying the geodesic equation, for which \(\gamma\left(0\right)=p\) and \(\gamma'\left(0\right)=\xi^{\alpha}\). Such a geodesic satisfies the following properties:

\begin{itemize}
\item it has one specified end point;

\item it has been extended as far as possible in some direction from the end point; 

\item it has a given affine parameter; and

\item the affine parameter vanishes at the end point and is positive elsewhere on the curve.
\end{itemize}

The \(g\) boundary is constructed from the subset \(G_I\subset G\), where \(G_I\) is obtained via the construction of a field \(\varphi\) that identifies points of \(G\) with the total affine length of the associated geodesics uniquely determined by those points. Elements of \(G_I\) are then those points \(\left(p,\xi^a\right)\in G\) such that \(\varphi\left(p,\xi^{\alpha}\right)\) is finite.

One constructs a \(9\)-dimensional manifold \(H=G\times\left(0,\infty\right)\). Define the folowing subsets \(H_0\) and \(H_+\) of \(H\) as 
\begin{eqnarray}\label{a01}
H_0&=&\lbrace{\left(p,\xi^{\alpha},d\right)|\varphi\left(p,\xi^{\alpha}\right)=d\rbrace},\\
H_+&=&\lbrace{\left(p,\xi^{\alpha},d\right)|\varphi\left(p,\xi^{\alpha}\right)>d\rbrace}.
\end{eqnarray}
Then there is a natural map \(\Psi:H_+\longrightarrow M\), that assigns to each point \(\left(p,\xi^{\alpha},d\right)\in H_+\) a point \(p'\in M\) obtained by moving a distance \(d\) along the geodesic uniquely associated to \(\left(p,\xi^{\alpha}\right)\). One defines a topology on \(G_I\) as follows: Let \(O\) be an open set of \(M\), and define open subsets \(S\left(O\right)\) as consisting of those points \(\left(p,\xi^{\alpha}\right)\in G_I\) such that there exists an open set \(U\subseteq H\) containing the point \(\left(p,\xi^{\alpha},\varphi\left(p,\xi^{\alpha}\right)\right)\) of \(H_0\), and \(\Psi\left(U\cap H_+\right)\subseteq O\).

Given two open sets \(O_1,O_2 \) in \(M\), it can be shown that \(S\left(O_1\right)\cap S\left(O_2\right)=S\left(O_1\cap O_2\right)\), \cite{ger1}. These open sets therefore form the basis of a topology on \(G_I\). Equivalence classes of elements of \(G_I\) can now be constructed by requiring that two elements \(\alpha,\beta \in G_I\) are equivalent if they always appear in the same open set. The collection of such equivalence classes forms the \(g\) boundary, where the induced topology on the \(g\) boundary is the quotient topology, and is \(T_0\) separable, \cite{ger1,hock1}. However, the ideal separation property would be the \(T_2\) separation property. Geroch did provide a recipe for constructing a \(g\) boundary that is \(T_2\) separable, but this requires the construction of all \(T_2\) separable equivalence relations on \(G_I\), which in general would not be feasible.

A new manifold \(M^*\), called the spacetime with \(g\) boundary, can now be constructed from the disjoint union of the spacetime manifold and the \(g\) boundary, i.e., \(M^*=M\sqcup g\). A topology on \(M^*\) can be described as follows: A subset of the form \(\left(O,U\right)\), where \(O\) is an open set of \(M\) and \(U\) is an open set of \(g\), will be called open in \(M^*\) if \(S\left(O\right)\supset U\). These open sets on \(M^*\) form a basis for a topology on \(M^*\). It can be checked that an intersection of two such open sets is open. If \(U=S\left(O\right)\), then \(\left(O,U\right)\) will be called a \textit{full open} set of \(M^*\). For more details and discussions see reference \cite{ger1}, from which we have also adopted the notation for open sets of the topology on the manifold completion throughout the rest of the paper.

It is important to note that Geroch's original construction only considered the set \(G_I\) of incomplete geodesics, and only allows for the identification of singular boundary points. Since we are interested in relating the \(g\) and \(a\) boundaries, in what follows we will consider the set \(G\) which allows one to treat a wider class of curves, and, hence, of boundary points. This allows for a definition of open sets on \(G\), which we shall exploit in order to establish a relationship with the \(a\) boundary (see section \ref{03})

\subsection{The \(a\) boundary}

The construction of the \(a\) boundary relies on the existence of open embeddings into manifolds of the same dimension, or \textit{envelopments}, \cite{ss1}. An advantage of the \(a\) boundary construction is that it can be applied to any manifold \(M\) and it is independent of both the affine connection on \(M\) and the chosen family of curves in \(M\). If we specify a family of curves \(\mathcal{C}\) in \(M\), satisfying the bounded parameter property (to be formally defined in section \ref{03}) then the \(a\) boundary points can be classified as \textit{regular, points at infinity, unapproachable points}, or \textit{singularities}. This subsection summarizes arguments developed in reference \cite{ss1}, unless otherwise cited.

Let \(M\) be a spacetime manifold, and let \(\phi:M\longrightarrow \hat{M}\) be an envelopment of \(M\) into \(\hat{M}\) where the dimension of \(M\) is the same as the dimension of \(\hat{M}\).
 
\begin{definition}\label{def1}
A \textbf{boundary point} of \(M\) is a point \(p\in\partial_{\phi}\left(M\right)\) in the topological boundary of \(\phi\left(M\right)\subseteq \hat{M}\). A \textbf{boundary set} is a non-empty subset \(B\) of \(\partial_{\phi}M\), comprised of boundary points.
\end{definition}

Let \(\phi':M\longrightarrow \hat{M}'\) be a second envelopment of \(M\) into \(\hat{M}'\), and let \(B'\subseteq \partial_{\phi'}M\). We define a \textit{covering relation} as follows:
\begin{definition}\label{def2}
A boundary set \(B\subseteq \partial_{\phi}M\) in \(\hat{M}\)  \textbf{covers} a boundary set \(B'\subseteq \partial_{\phi'}M\) in \(\hat{M}'\) if for every open neighborhood \(U\) of \(B\) in \(\hat{M}\), there exists an open neighborhood \(U'\) of \(B'\) in \(\hat{M}'\) such that
\begin{eqnarray}\label{a02}
\phi\circ\phi'^{-1}\left(U'\cap \phi'\left(M\right)\right)&\subseteq & U.
\end{eqnarray}
\end{definition}

\begin{definition}\label{def3}
A boundary set \(B\) is \textbf{equivalent} to a boundary set \(B'\) if \(B\) covers \(B'\) and \(B'\) covers \(B\).
\end{definition}
The covering relation defines an equivalence relation on the set of all boundary sets induced by all possible envelopments of \(M\). An equivalence class \(\left[B\right]\) of boundary sets is called an \textit{abstract boundary set}.
\begin{definition}\label{def4}
An \textbf{abstract boundary point} is an abstract boundary set that has a singleton \(p\) as a representative element. The set of all abstract boundary points is called the \textbf{abstract boundary} or simply the \(a\) boundary.
\end{definition}

Let \(O\) be an open set in \(M\) and let \(B\in \partial_{\phi}M\) be a boundary set in the topological boundary of an envelopment \(\phi:M\longrightarrow \hat{M}\) of \(M\) into \(\hat{M}\). A boundary point \(p\in B\) (respectively, a boundary set \(B\)) is attached to the open set \(O\) of \(M\) if for every open neighborhood \(U\) in \(\hat{M}\) of \(p\) (resp., of \(B\)), we have that \(U\cap \phi\left(O\right)\neq  \varnothing\). For a boundary set to be attached to the open set \(O\), we require at least one boundary point \(p\in B\) to be attached to \(O\) . An abstract boundary point \(\left[p\right]\) is attached to an open set \(O\subseteq M\) if the boundary point \(p\) is attached to \(O\) (see the reference \cite{ss2} for more details on the attached point topology).

Again, one wants a topology on the disjoint union of the spacetime manifold \(M\) and the abstract boundary \(a\): \(\bar{M}\equiv M\sqcup a\). Subsets of \(\bar{M}\) of the form \(\left(O\cup B,C\right)\), where \(O\) is a nonempty open set of \(M\), \(B\) is the set of all abstract boundary points which are attached to \(O\), and \(C\) is some subset of the abstract boundary, and where the collection of every \(C\) set is the set of all subsets of the abstract boundary, will be called open in \(\bar{M}\). These open sets form a basis for a topology on \(\bar{M}\), which is Hausdorff, \cite{ss2}. This is known as the attached point topology, and the open sets of the topology induced on \(a\) from the attached point topology are precisely the \(C\) sets.

By construction, the identification of an abstract boundary point \(\left[\hat{p}\right]\), eg. a singularity, is completely determined by the open set of \(M\) to which \(\left[\hat{p}\right]\) is attached. This is one of the primary motivations for boundary constructions. What this means is that one has an envelopment \(\phi\) of a manifold \(M\) into, say, \(\hat{M}\), with \(\hat{p}\in \partial_{\phi}M\) and \(\hat{p}\) is attached to \(O\) (which implies \(\left[\hat{p}\right]\) is attached to \(O\)). Then the image \(\phi\left(O\right)\) extends (in \(\hat{M}\)) to \(\hat{p}\in \hat{M}\). One therefore has knowledge as to an open set in \(M\) that is close to the singularity, and thus a means of locating them.

\section{A new construction}\label{03}

As mentioned in section \ref{nnjjt}, we intend our construction to be applicable to a wider class of curves - those satisfying what is called the bounded parameter property - as in the case of the \(a\) boundary. In this section we first briefly discuss the class of curves satisfying the bounded parameter property, of which geodesics form a subclass. We then proceed to the construction of a Geroch-like boundary, and its relationship with the \(a\) boundary, before concluding with a few illustrative examples.

\subsection{The bounded parameter property}

The \(a\) boundary considers a broad class of curves satisfying the \textit{bounded parameter property} (b.p.p), \cite{ss1}.

Let \(\gamma:\left[a,b\right)\longrightarrow M\) (with \(a<b\leq \infty\)) be a parametrized and regular (the tangent vector \(\dot{\gamma}\) is nowhere vanishing on the interval \(\left[a,b\right)\)) curve in \(M\), with \(\gamma\left(a\right)=p\) as the starting point of \(\gamma\). A curve \(\gamma':\left[a',b'\right)\longrightarrow M\) is a subcurve of \(\gamma\) if \(a\leq a' < b' \leq b\) and \(\gamma'=\gamma|_{\left[a',b'\right)}\). If \(a=a'\) and \(b>b'\) we say that \(\gamma\) is an extension of \(\gamma'\).

\begin{definition}[Change of parameter]\label{def5}
A \textbf{change of parameter} is a monotone increasing \(C^1\) function
\begin{eqnarray*}
s:\left[a,b\right)\longrightarrow \left[a',b'\right),
\end{eqnarray*}
which maps \(a\) to \(a'\) and \(b\) to \(b'\), and \(\frac{ds}{d\lambda}>0\) for \(\lambda\in \left[a,b\right)\). The curve \(\gamma'\) is obtained from \(\gamma\) via the change of parameter \(s\) if the following diagram commutes:
\begin{eqnarray}\label{a03}
\begin{tikzcd}
\left[a,b\right) \arrow["s", rd] \arrow[r, "\gamma"] & M \\
                                 & \left[a',b'\right) \arrow[u, "\gamma'"].
\end{tikzcd}
\end{eqnarray}
\end{definition}

\begin{definition}[Bounded parameter property (b.p.p.)]\label{def6}
A family of parametrized curves \(\mathcal{C}\) in \(M\) is said to have the bounded parameter property if
\begin{itemize}
\item for any \(p\in M\), \(\exists\) at least one \(\gamma\in \mathcal{C}\) such that \(\gamma\left(\lambda\right)=p\) for some \(\lambda\in \left[a,b\right)\);

\item if \(\gamma\in \mathcal{C}\), so is every subcurve \(\gamma'\) of \(\gamma\); and 

\item for any pair \(\gamma,\gamma'\in \mathcal{C}\) for which there exists an \(s\) such that the diagram in (\ref{a03}) commutes, we have either the parameter on both \(\gamma\) and \(\gamma'\) is bounded, or the parameter on both \(\gamma\) and \(\gamma'\) is unbounded.
\end{itemize}
\end{definition}

Families of curves satisfying the b.p.p. include geodesics with affine parameter, differentiable curves with generalized affine parameter and timelike geodesics parametrized by proper time. This generality carries over to the generality of the \(a\) boundary construction, as well as our our construction. As with the \(a\) boundary, our construction applies to a smooth manifold of any dimension endowed with an affine connection and b.p.p. curves. While this is still modest compared to the \(a\) boundary which exists for any manifold independent of an affine connection and a family of curves, our construction could, in principle, equally be applied to theories such as Einstein-Cartan, Kaluza-Klein, Yang-Mills, etc., though this is not the subject of this current work. For the remainder of this paper all curves will satisfy the b.p.p., unless otherwise stated. 

We conclude this subsection with several standard definitions of utility in the sequel.

\begin{definition}[Limit point]
Let \(\gamma:\left[a,b\right)\longrightarrow M\) be a curve. A point \(\hat{p}\) is a limit point of \(\gamma\) if there exists an increasing infinite sequence of real numbers \(t_i\in\left[a,b\right)\) with \(t_i\rightarrow b\) such that \(\gamma\left(t_i\right)\rightarrow \hat{p}\). This implies that for every subcurve \(\gamma':\left[a',b\right)\longrightarrow M\) with \(a\leq a' < b\), \(\gamma\left(t\right)\) enters every neighborhood \(U\) of \(\hat{p}\). 
\end{definition}

\begin{definition}[Endpoint]
Let \(\gamma:\left[a,b\right)\longrightarrow M\) be a curve. A point \(\hat{p}\) is an endpoint of \(\gamma\) if \(\gamma\left(t\right)\rightarrow \hat{p}\) as \(t\rightarrow b\). For Hausdorff manifolds, \(\hat{p}\) is unique.
\end{definition}

\begin{definition}[Approach by a curve]
Let \(\phi: M\longrightarrow \hat{M}\) be an envelopment of \(M\) into \(\hat{M}\), and let \(\gamma:\left[a,b\right)\longrightarrow M\) be a curve in \(M\). We say \(\gamma\) \textbf{approaches} a boundary set \(B\in \partial_{\phi}M\) if there exists a \(\hat{p}\in B\) such that \(\hat{p}\) is a limit point of \(\phi\circ\gamma\).
\end{definition}

\subsection {A Geroch-like boundary}\label{sub31}

In this subsection we define open sets on the reduced tangent bundle \(G\) making use of the attachment relation from the \(a\) boundary and a natural topology on the tangent bundle. This allows us to define an equivalence relation on \(G\), and the collection of the associated equivalence classes defines a Geroch-like boundary which we shall denote \(\tilde{g}\).  

In what follows, all manifolds are assumed to be smooth, Hausdorff, connected, and paracompact. The condition of paracompactness ensures that all manifolds considered herein are metrizable.

Given a manifold \(M\), we recall that the tangent bundle \(TM\) on \(M\) can be given a natural topology that is Hausdorff (for example, see reference \cite{curm1}). Consider this topology on \(TM\) with \(O\) being an open subset of \(M\), and let \(\gamma_p\) be the curve associated with a point \(\left(p,\xi^{\alpha}\right)\in G\subset TM\) (whenever we write \textit{\(\gamma_p\) associated with the point \(\left(p,\xi^{\alpha}\right)\)}, we will mean a b.p.p. curve \(\gamma_p:\left[a,b\right)\longrightarrow M\) starting at the manifold point \(\gamma\left(a\right)=p\)). Suppose \(\phi:M\longrightarrow \hat{M}\) is an envelopment from \(M\) into \(\hat{M}\), and let \(\left[\hat{p}\right]\) be an abstract boundary point attached to \(O\), with \(\hat{p}\) in the topological boundary \(\partial_{\phi}M\). We define open sets on \(G\) as follows:

\begin{eqnarray}\label{to1}
\begin{split}
S\left(O\right)&=\{\left(p,\xi^{\alpha}\right)\in G\  | \ \phi\left(\gamma_p\left(\lambda\right)\right)\in \phi\left(O\right), \forall \ \lambda\in \left[a,b\right) \ \text{and} \\
&\ \ \ \ \ \phi\left(\gamma_p\right)\supset\phi\left(\gamma'_p\right)\subset\phi\left(O\right)\cap N_{\hat{p}}, \ \text{for subcurve}\ \phi\left(\gamma'_p\right) \\
&\ \ \ \ \ \text{of}\  \phi\left(\gamma_p\right)\ \text{for all open neighborhoods}\  N_{\hat{p}}\  \text{of}\  \hat{p}, \ \text{over}\\
&\ \ \ \ \ \text{all possible envelopments}\ \phi\ \text{of}\ M\}.
\end{split}
\end{eqnarray}
Of course, \(S\left(O\right)\) is empty if no boundary set is attached to \(O\). It is also easily seen that \(S\left(O_1\right)\cap S\left(O_2\right)=S\left(O_1\cap O_2\right)\). We therefore have that the \(S\left(O\right)\) sets form a basis for a topology on \(G\). (See figure \ref{fig:a} for a depiction of the attachment of \(\hat{p}\) to \(O\) and the portion of the curve in the intersection.)

\begin{figure}
\def\svgwidth{14cm}
\def\svgheight{12cm}
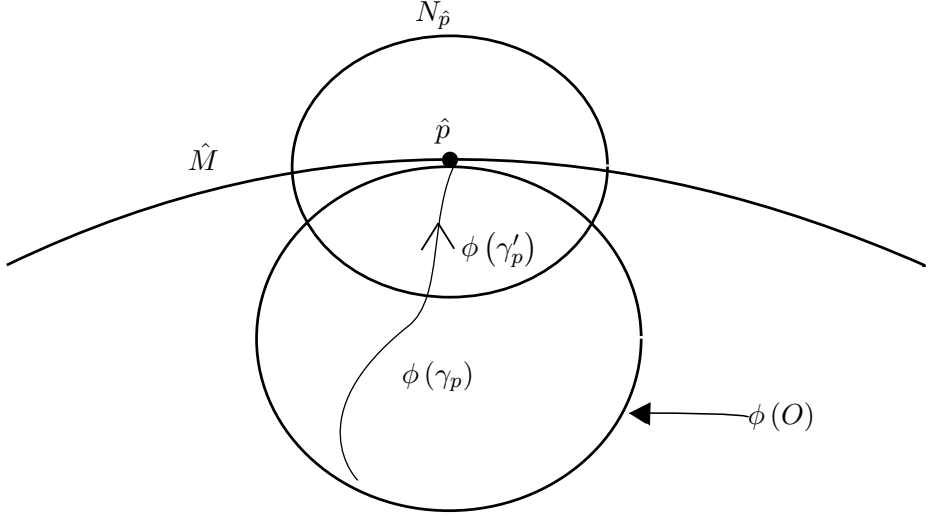
\caption{Depiction of the open set \(\phi\left(O\right)\) containing \(\phi\left(\gamma_p\right)\) to which \(\hat{p}\) is attached, and the portion \(\gamma'_p\) of \(\gamma_p\) lying in the intersection \(\phi\left(O\right)\cap N_{\hat{p}}\).}
\label{fig:a}
\end{figure}

\begin{definition}\label{def05}
Let \(\left(p,\xi^{\alpha}\right)\) and \(\left(q,\xi^{\beta}\right)\) be points in \(G\). We define an \textbf{equivalence relation}, denoted \(E_1\), on \(G\) as follows: \(\left(p,\xi^{\alpha}\right)\sim \left(q,\xi^{\beta}\right)\) iff the curves associated with \(\left(p,\xi^{\alpha}\right)\) and \(\left(q,\xi^{\beta}\right)\) approach the same abstract boundary point. 
\end{definition}
The proof that \(E_1\) is indeed an equivalence relation follows immediately from the definition of an abstract boundary point.

We stress that the use of the abstract boundary point in definition (\ref{def05}) means that we are considering equivalence classes of boundary sets of envelopment via definition (\ref{def3}). We also note that, if two curves \(\gamma_1\) and \(\gamma_2\) both start at the same point, then from (\ref{to1}), it is clear that they will both approach the same abstract boundary point. 

\begin{definition}\label{def155}
The collection of all equivalence classes in \(G\) under the equivalence relation \(E_1\) will be called the \(\tilde{g}\) boundary. 
\end{definition}

The \(\tilde{g}\) boundary naturally inherits the coinduced topology from \(G\) under the continuous canonical quotient map \(q\) from \(G\) to \(\tilde{g}\) (In section \(3.3\) we make use of this as part of our discussion about explicitly embedding the \(\tilde{g}\) boundary  into the \(a\) boundary) \cite{munk1}. We emphasize that our construction relies on the existence of open embeddings. This is not an overly restrictive requirement as all theories specified via a metric in some coordinate system (eg. general relativity, Einstein-Gauss-Bonnet, etc.), have an implicit envelopment. (See the discussion in the conclusion of reference \cite{ss1}.) \cite{ger1}, noted that there is no unique way of defining the \(g\) boundary, and that his particular construction was ad hoc. However, under the condition that envelopments of the spacetime manifold exist, the \(a\) boundary provides us with a natural way of defining open sets on \(G\). It also allows for the definition of a desired quotient on \(G\) via definition (\ref{def155}).

We now presents some results on properties of the \(\tilde{g}\) boundary.

\begin{theorem}\label{propp2}
The \(\tilde{g}\) boundary is Hausdorff separable in the coinduced topology from \(G\).
\end{theorem}
\begin{proof}
Let \(\sigma_1\) and \(\sigma_2\) be two points in \(\tilde{g}\), and let 
\begin{eqnarray*}
\bigcup\limits_{i\in I}S\left(O_i\right)\  \text{and}\  \bigcup\limits_{j\in J}S\left(O_j\right) 
\end{eqnarray*}
be open neighborhoods of \(\sigma_1\) and \(\sigma_2\) respectively. Choose open subsets of \(M\) such that \(O_m\cap O_n=\varnothing\), for any pair \(m,n\) such that \(m\in I\) and \(n\in J\). Now fix an envelopment \(\phi:M\longrightarrow \hat{M}\). Curves associated with the points \(\sigma_1\) and \(\sigma_2\) approach different abstract boundary points, and therefore approach different boundary points (or sets) in \(\partial_{\partial}M\), say \(\hat{p}_1\) and \(\hat{p}_2\) attached to \(O_m\) and \(O_n\) respectively (for fixed \(m\) and \(n\)), that are not equivalent under the abstract boundary equivalence. In the topology induced by \(\hat{M}\) on \(\partial_{\phi}M\), we can always choose disjoint open neighborhoods
\(N_{\hat{p}_1}\) and \(N_{\hat{p}_1}\) of \(\hat{p}_1\) and \(\hat{p}_2\) respectively. Since \(O_m\cap O_n=\varnothing\), and since curves in \(O_m\) and \(O_n\) enter and stay in \(O_m\cap N_{\hat{p}_1}\) and \(O_n\cap N_{\hat{p}_2}\) respectively, \(S\left(O_m\right)\cap S\left(O_n\right)=\varnothing\). We therefore have that
\begin{eqnarray*}
\left(\bigcup\limits_{i\in I}S\left(O_i\right)\right)\cap\left(\bigcup\limits_{j\in J}S\left(O_j\right)\right)=\varnothing,
\end{eqnarray*}
and hence concludes the proof.
\end{proof}
We note that \(G\) is not Hausdorff separable in the topology with basis given by the open sets \(S\left(O\right)\). This is because any two points \(\left(p,\xi^{\alpha}\right),\left(p,\xi^{\beta}\right)\in G\) will always lie in the same open set since curves associated to the two points start at the same manifold point and thus remain in the open set of the manifold containing the point. However, we have shown that \(\tilde{g}\) is Hausdorff separable since under the equivalence relation \(E_1\) these points are identified.
\begin{theorem}\label{chee1}
Let \(M^*\) be the manifold formed by the completion of the spacetime manifold \(M\) with \(\tilde{g}\): \(M^*\equiv M\sqcup \tilde{g}\), and let a basis for a topology on \(M^*\) be given by open subsets of the form 
\begin{eqnarray}\label{ffgpapee}
\left(O,\tilde{U}=\bigcup\limits_{i\in I}S\left(O_i\right)\right),
\end{eqnarray}
where \(\tilde{U}\) are the open sets of \(\tilde{g}\), coinduced by the open sets on \(G\), and \(O=O_k\) for some \(k\in I\) is an open set in \(M\). Then \(M^*\) is Hausdorff. 
\end{theorem}

\begin{proof}

That \(M^*\) is Hausdorff follows from the fact that the disjoint union preserves the Hausdorff property: for points \(p\in M\) and \(\sigma\in\tilde{g}\), open neighborhoods \(O_k\ni p\) and  
\begin{eqnarray*}
\tilde{U}=\bigcup\limits_{i\in I}S\left(O_i\right)\ni \sigma
\end{eqnarray*}
are disjoint.
\end{proof}

One may notice that, for our choice of basis for a topology on \(M^*\), \eqref{ffgpapee}, the open set \(O\) of \(M^*\) restricted to \(M\) is simply required to be a choice of one of the \(O_i's\) defining the open set \(\tilde{U}\) of \(M^*\) restricted to \(\tilde{g}\). This may be done because boundary points (sets) attached to all of the \(O_is\) are identified under the abstract boundary equivalence.

We stress again that, in the entirety of this subsection's construction we have considered b.p.p. curves \(\gamma_p\) associated to \(\left(p,\xi^{\alpha}\right)\in G\). In general they need not be geodesics. However, an explicit reduction to a construction more similar to that of Geroch is obtained by insisting that they are geodesics. This restricted case remains a generalization of Geroch's \(g\) boundary in the sense that we have an additional equivalence relation \(E_1\), which is necessary in order to yield \(T_2\) separability.

The key notion here is a rephrasing of the \(a\) boundary using b.p.p.curves explicitly related to the tangent bundle. This is what facilitates the understanding of the relationship between the \(a\) boundary and other boundary constructions, in this case, a construction similar to the \(g\) boundary. This relationship is considered in more detail in the next subsection.

\subsection{Relation to the \(a\) boundary}
In this section we now show how the \(\tilde{g}\) boundary can be embedded into the \(a\)  boundary via explicit mappings. We first define an equivalence relation on \(\phi\left(M\right)\). 

\begin{definition}\label{def06}
Let \(\phi \left(p_1\right),\phi \left(p_2\right)\in \phi\left(M\right)\). The \textbf{equivalence relation} on \(\phi\left(M\right)\), denoted \(E_2\), is defined as follows: \(\phi \left(p_1\right)\sim \phi \left(p_2\right)\) iff 
\begin{itemize}
\item[a.] \(\phi \left(p_1\right)=\phi\circ \pi \left(p_1,\xi^{\alpha}_1\right)\) and \(\phi \left(p_2\right)=\phi\circ \pi \left(p_2,\xi^{\alpha}_2\right)\), where the \(\left(p_i,\xi_i^{\alpha}\right)\) are in open sets of \(G\), and 

\item[b1.] \(\phi \left(p_1\right)\) and \(\phi \left(p_2\right)\) lie on the same curve, or

\item[b2.] \(\phi \left(p_1\right)\) and \(\phi \left(p_2\right)\) lie on curves approaching the same abstract boundary point.
\end{itemize}
\end{definition}
This yields equivalence classes of curvess with limit points being the same abstract boundary point, the collection \(\phi\left(M\right)/_{E_2}\), which we shall denote by \(\phi\left(M\right)_{geo}\). It is not difficult to show that \(E_2\) is indeed an equivalence relation. Reflexivity and symmetry follow immediately from the definition. To show transitivity, let \(\phi\left(p_1\right)\sim \phi\left(p_2\right)\). Then either \(\phi\left(p_1\right)\) and \(\phi\left(p_2\right)\) lie on the same curve \(\gamma\) or \(\phi\left(p_1\right)\) and \(\phi\left(p_2\right)\) lie on curves approaching the same abstract boundary point. Suppose \(\phi\left(p_1\right)\) and \(\phi\left(p_2\right)\) both lie on \(\gamma\). If \(\phi\left(p_2\right)\) and \(\phi\left(p_3\right)\) lie on \(\gamma\), then \(\phi\left(p_1\right)\sim \phi\left(p_3\right)\). Otherwise, \(\phi\left(p_3\right)\) lies on some curve \(\tilde{\gamma}\) such that \(\tilde{\gamma}\) approaches the same abstract boundary point as \(\gamma\). Since \(\phi\left(p_1\right)\) lies on \(\gamma\), this would imply \(\phi\left(p_1\right)\sim \phi\left(p_3\right)\). Now suppose \(\phi\left(p_1\right)\) and \(\phi\left(p_2\right)\) lie on curves \(\gamma\) and \(\tilde{\gamma}\) respectively, both approaching the same abstract boundary point. If \(\phi\left(p_3\right)\) lies on either of \(\gamma\) or \(\tilde{\gamma}\), then \(\phi\left(p_3\right)\sim \phi\left(p_1\right)\) or \(\phi\left(p_3\right)\sim \phi\left(p_2\right)\) which would imply \(\phi\left(p_3\right)\sim \phi\left(p_2\right)\) or \(\phi\left(p_3\right)\sim \phi\left(p_1\right)\) respectively. Otherwise, \(\phi\left(p_3\right)\) lies on some curve \(\hat{\gamma}\) (different from \(\gamma\) and \(\tilde{\gamma}\)) approaching the same abstract boundary point as \(\gamma\) and \(\tilde{\gamma}\), which would imply that \(\phi\left(p_3\right)\sim \phi\left(p_1\right)\), \(\phi\left(p_1\right)\sim \phi\left(p_2\right)\) and \(\phi\left(p_3\right)\sim \phi\left(p_2\right)\).

Let the map
\begin{eqnarray}\label{to2}
q:G\longrightarrow \tilde{g},
\end{eqnarray}
be the canonical quotient map defined by
\begin{eqnarray*}
\left(p,\xi^{\alpha}\right)\mapsto \left[\left(p,\xi^{\alpha}\right)\right]_{E_1},
\end{eqnarray*}
which sends points in \(G\) to its equivalence class under the equivalence relation \(E_1\). Define a map
\begin{eqnarray}\label{to3}
\kappa:G\longrightarrow \phi\left(M\right)_{geo},
\end{eqnarray}
by 
\begin{eqnarray*}
\left(p,\xi^{\alpha}\right)\mapsto \left[\phi\left(p\right)\right]_{E_2},
\end{eqnarray*}
which maps a point \(\left(p,\xi^{\alpha}\right)\) of \(G\) to the equivalence class (under the equivalence relation \(E_2\)) containing the image of the associated curve under the composition \(\phi\circ\pi\). (The set \(\phi\left(M\right)_{geo}\) is just the collection of the equivalence classes \(\left[\phi\left(p\right)\right]_{E_2}\).) The map \(q\) is a quotient map and thus has the natural inverse map, \(q^{-1}\), which sends an equivalence class to the set of its elements. Clearly \(\kappa\) is constant on the set \(q^{-1}\left(\left[\left(p,\xi^{\alpha}\right)\right]_{E_1}\right)\), for \(\left[\left(p,\xi^{\alpha}\right)\right]_{E_1}\in \tilde{g}\), since all elements in \(q^{-1}\left(\left[\left(p,\xi^{\alpha}\right)\right]_{E_1}\right)\) are sent to \(\left[\phi\left(p\right)\right]_{E_2}\). The map \(\kappa\) thus induces a map (see theorem \(22.2\) of reference \cite{munk1})
\begin{eqnarray}\label{to4}
r:\tilde{g}\longrightarrow \phi\left(M\right)_{geo},
\end{eqnarray}
such that \(r\circ q=\kappa\). Hence, the diagram
\begin{eqnarray}\label{to5}
\begin{tikzcd}
G \arrow[rd, "\kappa"] \arrow[r, "q"] & \tilde{g} \arrow[d, "r"] \\
&\phi\left(M\right)_{geo}.
\end{tikzcd}
\end{eqnarray}
commutes. The map \(r\) sends points of \(\tilde{g}\) to the appropriate equivalence class under the equivalence relation \(E_2\).

We now introduce the notion of the limit operator, \cite{lo1,lo2,lo3}, which allows us to attach limit points/endpoints to curves from \(\phi\left(M\right)\). Let \(X\) be any set and let \(S\left(X\right)\) denote the set of sequences in \(X\). Let \(P\left(X\right)\) denote the set of parts of \(X\). One defines the limit operator as a map \(l:S\left(X\right) \longrightarrow P\left(X\right)\) which sends sequences in \(S\left(X\right)\) to their limit points in \(P\left(X\right)\), satisfying the compatibility condition of subsequences: if \(\sigma_1,\sigma_2\in S\left(X\right)\), and \(\sigma_2\) is a subsequence of \(\sigma_1\), then \(l\left(\sigma_1\right)\subset l\left(\sigma_2\right)\).

As has been mentioned, all manifolds considered here are metrizable. Let \(t_i\) be an increasing infinite sequence of real numbers. Then a curve \(\phi\left(\gamma_p\right)\in \phi\left(M\right)\) can be written as the sequence \(\phi\left(\gamma_p\left(t_i\right)\right)\). The limit operator
\begin{eqnarray*}
L:\phi\left(M\right)_{geo} \longrightarrow \phi\left(M\right)^l_{geo},
\end{eqnarray*}
where \(\phi\left(M\right)^l_{geo}\) consists of equivalence classes of limit points of the sequences in \(\phi\left(M\right)\), takes equivalence classes of curves in \(\phi\left(M\right)_{geo}\) to equivalence classes of limit points in \(\partial_{\phi}M\) under \(E_2\). The equivalence classes are elements of the abstract boundary. The composition 

\begin{eqnarray}\label{to500}
L\circ r:\tilde{g} \longrightarrow \phi\left(M\right)^l_{geo}
\end{eqnarray}
given by
\begin{eqnarray}\label{to6}
\left[\left(p,\xi^{\alpha}\right)\right]_{E_1} \mapsto \left[B\right],
\end{eqnarray}
sends the \(\tilde{g}\) boundary point \(\left[\left(p,\xi^{\alpha}\right)\right]_{E_1}\) to the associated equivalence class \(\left[B\right]\), for some boundary set \(B\) in \(\partial_{\phi}M\), of limit points in \(\phi\left(M\right)^l_{geo}\). As \(B\) is a representative set of an abstract boundary point (from the definition of \(\phi\left(M\right)_{geo}\) in (\ref{to3})), we can write the composite \(L\circ r\) as the map

\begin{eqnarray*}
L\circ r:\tilde{g} \longrightarrow \bar{a},
\end{eqnarray*}
defined by
\begin{eqnarray}\label{to7}
\left[\left(p,\xi^{\alpha}\right)\right]_{E_1} \mapsto \left[\hat{p}\right],
\end{eqnarray}
where \(\left[\hat{p}\right]\ni B\) and \(\bar{a}\subseteq a\) is some subset of the abstract boundary \(a\). The map \(L\circ r\) is a bijection, and gives us a desired map from \(\tilde{g}\) to \(a\).  

\begin{theorem}\label{theo8}
The composite map \(L\circ r\) is a homeomorphism from \(\tilde{g}\) onto its image in \(a\).
\end{theorem}

\begin{proof}
Clearly, \(L\circ r\) is a bijection onto its image \(L\circ r\left(\tilde{g}\right)\): any two equivalent points in \(\left[\hat{p}\right]\) are endpoints/limit points of curves associated with equivalent points in \(G\) (under \(E_1\)). Since \(\kappa\) is a quotient map, \(\kappa\) is continuous. By Theorem \(22.2\) of reference \cite{munk1}, \(r\) is also continuous. The limit operator \(L\) can be considered as a map sending collection of points in \(\phi\left(M\right)\) (in this case those lying on a curve) to a point in the boundary of \(\phi\left(M\right)\) in \(\hat{M}\) (in this case the limit point of the curve on which those points lie). Then \(L\) is a quotient map and therefore continuous. Since the composition of continuous maps is continuous, \(L\circ r\) is therefore continuous.

Since the boundary of a bounded set is compact, and the spacetime completion can be written as the union of \(M\) and \(\tilde{g}\), we know that \(\tilde{g}\) is compact. We also know that the abstract boundary \(a\) (and every subset of \(a\)) is Hausdorff. By Theorem (26.6) in reference \cite{munk1}, a continuous bijection from a compact space to a Hausdorff space is a homeomorphism. We therefore have that \(L\circ r\) is a homeomorphism, i.e. \(L\circ r\) embeds \(\tilde{g}\) into \(a\).
\end{proof}

That this embedding might be a proper embedding is an important point to note since there might also be \textit{unapproachable} boundary points (see reference \cite{ss1}) that are not approached by any curve. There might also be cases where the image \(L\circ r\left(\tilde{g}\right)=\bar{a}\) coincides with the \(a\) boundary, and future work could consider under what conditions this happens. We stress that this construction holds for any family of b.p.p. curves (not just geodesics) yielding a generalization of Geroch's \(g\) boundary. 

Via its deep relationship to the \(a\) boundary and the manner in which that boundary construction may be used (given some metric) to classify boundary points as \textit{regular/singular/point at infinity/unapproachable}, our construction similarly allows for the classification of all approachable boundary points and not just singular ones, which could prove useful for future work where we might want to look at other boundary constructions built from the tangent bundle structure. We emphasize that, for the purposes of this paper, the focus is on maps of the form (\ref{to7}) since  we are primarily interested in curves as they approach their associated limit points.

\subsection{Simple illustrative examples}\label{rety}

Simple examples that illustrate the subtle differences between our construction and Geroch's original construction include the Schwarzschild spacetime and Misner's simplified version of the Taub-NUT spacetime. For the Schwarzschild spacetime, the \(\tilde{g}\) boundary is a single point, \(\alpha\), consisting of points in \(G\) whose associated geodesics all approach the surface \(r=0\), which is a point in the \(a\) boundary. By contrast, the \(g\) boundary is made up of two parts, each topologically \(\mathbb{S}^2\times \mathbb{R}\), for different approaches of the geodesics to \(r=0\), \cite{ger1}. 

For Misner's simplified version of the Taub-NUT spacetime, \(\tilde{g}\) consists of just one point, \(\beta\), which is the equivalence class consisting of those points in \(G\) associated with all geodesics that approach \(t=0\). In contrast, the \(g\) boundary consists of three points, \(\tilde{g}\), and two circles, \(C\) and \(C'\), \cite{ger1}, with a point on the circles representing those geodesics from both halves of the cylinder striking the circles at exactly one point.

Consider another simple example. Let \(\phi\) embed the unit interval \(\left(0,1\right)\) into \(\mathbb{R}\) via the inclusion map. The boundary set \(B=\lbrace{0,1\rbrace}\) is the boundary set of this envelopment. Now let a second envelopment, \(\phi'\), embed \(\left(0,1\right)\) in the unit circle via the map \(\theta=2\pi t\). The boundary points of the first envelopment are both identified with the boundary point \(0\) of the second envelopment and so \(0\) and \(1\) of the first envelopment are equivalent (under the equivalence relation defining the abstract boundary) to \(0\) of the second envelopment. The boundary set \(B\) is disconnected and so takes the discrete topology. Therefore, the singletons \(\lbrace{0\rbrace}\) and \(\lbrace{1\rbrace}\) are open sets containing each of the boundary points.

Suppose we are presented with just the envelopment \(\phi\). The boundary points \(0\) and \(1\) may wrongly be associated with different equivalence classes in the \(g\) boundary. With the knowledge of the second envelopment,  it becomes clear that they are equivalent (since they are both equivalent to \(0\) in the second envelopment). We thus have \(0\) and \(1\) as elements of an equivalence class in the \(\tilde{g}\) boundary.

Let us state in general terms a subtle point for considering not just a single envelopment. Suppose we have a manifold \(M\), and let \(\phi,\phi'\) be two envelopments of \(M\). Let \(\partial_{\phi}M=B\) be a regular boundary set (comprise of regular boundary points), and suppose \(\hat{p}\in \partial_{\phi'}M\) is an apparently essential singular boundary point which is equivalent to \(B\) under the abstract boundary equivalence. Then the point \(\left[\hat{p}\right]\in\tilde{g}\) of \(M\) contains \(\hat{p}\) and \(B\).

Now, suppose one is \textit{only} presented with the second envelopment \(\phi'\). One might wrongly conclude that the manifold \(M\) is singular. However, since \(B\) is regular and \(B\sim \hat{p}\), by the classification of the different types of boundary points, the singularity of \(\hat{p}\) is an artefact of the choice of envelopment; \(\left[\hat{p}\right]\) is not a singular point.

\section{Limitations}\label{gtr}

Relating the \(g\) and \(a\) boundaries as have been done in this paper comes with the inherited problem from the attached point topology of the \(a\) boundary: there is a detachment of manifold points from singular points in the topology on the completion of the spacetime manifold with \(\tilde{g}\). This was (in the \(a\) boundary context), in fact, the motivation of the construction of the \textit{strongly attached point topology} introduced by \cite{ss10}. We do not solve this problem in this work. It is outside the scope of the paper as was made clear in the objectives of this work. This, however, does not mean that there is not a potential solution using our approach.

As was noted earlier, the attachment relation depends on the intersection of all open neighborhoods of a boundary point (set) with the image of the open set of \(M\) in \(\hat{M}\) to which the boundary point (set) is attached, i.e. \(N\cap \phi(O)\neq \varnothing\), for all open neighborhoods \(N\subset \hat{M}\) of the boundary point (set). A boundary point (set) being strongly attached to an open set in \(M\) is a much stricter condition \cite{ss10}. In particular, it requires the existence of an open neighborhood \(N\subset \hat{M}\) of \(\hat{p}\) (\(B\)) such that \(N\cap \phi(M)\subset \phi(O)\). (Notice that strongly attached implies attached (this was proved in \cite{ss10}), but the converse does not hold in general.)

An important property of the strong attachment relation not possessed by the attachment relation is that \textit{if a boundary point \(\hat{p}\) is strongly attached to open sets \(O_1\) and \(O_2\) of \(M\), then \(O_1\cap O_2\neq \varnothing\), and \(\hat{p}\) is also strongly attached to the intersection \(O_1\cap O_2\)}.

Interpreted in terms of our approach, it is obvious that one would require a modification of the \(S\left(O\right)\) sets. Indeed, we see that we can have curves approaching boundary points, completely lying in the open set \(\phi\left(O\right)\) of the manifold, but outside the intersection \(N\cap\phi\left(M\right)\). Hence, the subsequent equivalence relations will have to be modified, which would lead to an embedding in a different way.

\section{Conclusion}\label{05}

The aim of this paper was to construct a Geroch-like boundary, \(\tilde{g}\), and establish a relationship between this \(\tilde{g}\) boundary and the \(a\) boundary. Our construction, in essence, has rephrased the \(a\) boundary in terms of the reduced tangent bundle. Moreover, we potentially have a means of generalizing other constructions built from the tangent bundle structure. Specializing to geodesics we obtain a narrower generalization of the Geroch boundary. In the process of obtaining the definition of the equivalence relation for our construction, all other identifications we attempted were non-Hausdorff. 

It appears that \(\tilde{g}\) is the only pertinent subset of the \(a\) boundary (of course this has to be rigorously established) with respect to the following two conditions of \cite{glw1}:

\begin{itemize}
\item[a.] For every incomplete geodesic \(\gamma_p\) associated with a point \(\left(p,\xi^{\alpha}\right)\in G\), there exists a point \(\hat{p}\in \partial_{\phi}M_I\) such that \(\phi\left(\gamma_p\right)\) limits to \(\hat{p}\),

\item[b.] the extension \(\overline{exp}:\bar{U}\longrightarrow \overline{\phi\left(M\right)}\) of the exponential map \(exp:U\longrightarrow \phi\left(M\right)\), which take points of \(G\) associated with incomplete geodesics of affine length exactly one to their endpoints/limit points in \(\overline{\phi\left(M\right)}\), is continuous, 
\end{itemize}
where the notation \(\partial_{\phi}M_I\) denotes the boundary set consisting of the collection of endpoints of incomplete geodesics in the image of \(M\) under \(\phi\), and the over bar denotes the topological closure. If that is the case, this might provide a route to establishing a relationship between the \(b\) and \(a\) boundaries, for example, given the relationship between Geroch's \(g\) boundary and the \(b\) boundary.

Now that we have a relationship between the \(\tilde{g}\) and the \(a\) boundary, we have identified the following as additional possible future work, in addition to the problems identified in Section \ref{gtr}:
\begin{itemize}
\item[1.] Can the causal, metric and differentiable structures defined on \(g\), \cite{ger1}, be extended to \(\tilde{g}\)?

\item[2.] If the answer to \(1.\) is yes, can our construction then provide a way to construct causal, metrical and differentiable structures on the \(a\) boundary?

\item[3.] Again, if the answer to \(2.\) is in the affirmative, what new information, if any, can we obtain from these structures on the \(a\) boundary?
\end{itemize}

\section*{Acknowledgements}

The authors would like to thank Professor Rituparno Goswami, Professor Dharmanand Baboolal and Professor Sunil D. Maharaj of the Department of Mathematics, UKZN, for useful suggestions and advice. We would also like to thank Professor Zurab Janelidze of the Department of Mathematical Sciences, Stellenbosch University, for useful discussions. AS acknowledges that this research was supported by the National Research Foundation (NRF), South Africa and the First Rand Bank.


\begin{thebibliography}{00}

\bibitem[Kruskal (1960)]{kru1} 
M. D. Kruskal, 
Maximal extension of Schwarzschild metric,
{\it Phys. Rev.}  \textbf{119} (1960), 1743.

\bibitem[Geroch (1968)]{ger1} 
R. Geroch, 
Local characterization of singularities in general relativity,
{\it J. Math. Phys.} \textbf{9} (1968), 450.

\bibitem[Scott \& Szekeres (1994)]{ss1} 
S. M. Scott and P. Szekeres, 
The abstract boundary -- a new approach to singularities of manifolds,
{\it J. Geom. Phys.} \textbf{13} (1994), 223.

\bibitem[Schmidt (1971)]{bb1} 
B. G. Schmidt, 
A new definition of singular points in general relativity,
{\it Gen. Rel. Grav.} \textbf{1} (1971), 269--280.

\bibitem[Geroch \textit{et al.} (1972)]{cc1} 
R. Geroch, E. H. Kronheimer, and R. Penrose, 
Ideal points in space-time,
{\it Proc. R. Soc. Lond. A} \textbf{327} (1972), 545--567.

\bibitem[Szekeres (1960)]{szk1} 
G. Szekeres, 
On the singularities of a Riemannian manifold,
{\it Publ. Mat. Debrecen} \textbf{7} (1960), 285.

\bibitem[Hawking \& Ellis (1973)]{haw1} 
S. Hawking and G. F. R. Ellis, 
\textit{The large scale structure of spacetime}
(Cambridge University Press, 1973).

\bibitem[Misner (1967)]{mis1} 
C. W. Misner, 
Relativity Theory and Astrophysics I: Relativity and Cosmology,
{\it Lectures in Applied Mathematics} \textbf{8} (1967), 160.


\bibitem[Flores (2007)]{lo4} 
J. Flores, 
The causal boundary of spacetimes revisited,
{\it Comm. Math. Phys.} \textbf{276} (2007), 611--643.


\bibitem[Ashley (2002)]{as1} 
M. Ashley, 
Singularity theorems and the abstract boundary construction,
PhD Thesis, The Australian National University, (2002).


\bibitem[Budic \& Sachs (1974)]{bs1} 
R. Budic and R. K. Sachs, 
Causal boundaries for general relativistic space‐times,
{\it J. Math. Phys.} \textbf{15} (1974), 1302--1309.

\bibitem[Szabados (1988)]{sz1} 
L. Szabados, 
Causal boundary for strongly causal spacetimes,
{\it Class. Quant. Grav.} \textbf{5} (1988), 121.

\bibitem[Harris (1998)]{sh1} 
S. Harris, 
Universality of the future chronological boundary,
{\it J. Math. Phys.} \textbf{39} (1998), 5427--5445.

\bibitem[Harris (2000)]{sh2} 
S. Harris, 
Topology of the future chronological boundary: universality for spacelike boundaries,
{\it Class. Quant. Grav.} \textbf{17} (2000), 551.

\bibitem[Geroch \textit{et al.} (1982)]{glw1} 
R. Geroch, L. Canbin and R. M. Wald, 
Singular boundaries of space–times,
{\it J. Math. Phys.} \textbf{23} (1982), 432--435.

\bibitem[Barry \& Scott (2011)]{ss2} 
R. A. Barry and S. M. Scott, 
The attached point topology of the abstract boundary for spacetime,
{\it Class. Quant. Grav.} \textbf{28} (2011), 165003.

\bibitem[Curiel (1999)]{ek1} 
E. Curiel, 
The analysis of singular spacetimes,
{\it Phil. Science} \textbf{66} (1999), S119--S145.

\bibitem[Curiel (2009)]{ek2} 
E. Curiel, 
General relativity needs no interpretation,
{\it Phil. Science} \textbf{76} (2009), 44--72.

\bibitem[Hocking \& Young (1961)]{hock1} 
J. G. Hocking and G. S. Young, 
\textit{Topology},
(Addison-Wesley, 1961).

\bibitem[Curtis \& Miller (1985)]{curm1} 
W. D. Curtis and F. R. Miller, 
\textit{Differential manifolds and theoretical physics},
(Academic Press, 1985).

\bibitem[Beem \textit{et al.} (1996)]{jpe1} 
J. K. Beem, P. E. Ehrlich and K. L. Easley, 
\textit{Global Lorentzian Geometry},
(New York: Marcel Dekker, Inc, P3, 1996).

\bibitem[Munkres (2000)]{munk1} 
J. Munkres, 
\textit{Topology, Second Edition}, 
(Prentice Hall, 2000).

\bibitem[Goreham (2004)]{lo1} 
A. Goreham, 
Sequential convergence in topological spaces,
{\it arXiv preprint} math/0412558 (2004).

\bibitem[Flores \textit{et al.} (2011)]{lo2} 
J. Flores, J. Herrera, and M. Sanchez, 
On the final definition of the causal boundary and its relation with the conformal boundary,
{\it Adv. Theo. Math. Phys.} \textbf{15} (2011), 991--1057.

\bibitem[Flores \textit{et al.} (2010)]{lo3} 
J. Flores, J. Herrera, and M. Sanchez, 
Gromov, Cauchy and causal boundaries for Riemannian, Finslerian and Lorentzian manifolds,
{\it arXiv preprint} arXiv:1011.1154 (2010).


\bibitem[Barry \& Scott (2014)]{ss10}
R. A. Barry and S. M. Scott,
The strongly attached point topology of the abstract boundary for spacetime,
{\it Class. Quant. Grav.} \textbf{31} (2014), 125004.


\end{thebibliography}
\end{document}